\documentclass[a4paper,conference]{IEEEtran}
\IEEEoverridecommandlockouts
\usepackage{cite}
\usepackage{amsmath,amssymb,amsfonts}
\usepackage{algorithmic}
\usepackage{algorithm}
\usepackage{graphicx}
\usepackage{textcomp}
\usepackage{xcolor}
\usepackage{amsthm}
\usepackage{mathtools,flexisym}
\usepackage{bm}
\usepackage{enumitem}
\usepackage{graphicx}
 \usepackage{verbatim}
\newtheorem{Example}{Example}
\usepackage{amsthm}
\def\BibTeX{{\rm B\kern-.05em{\sc i\kern-.025em b}\kern-.08em
    T\kern-.1667em\lower.7ex\hbox{E}\kern-.125emX}}
\DeclareMathOperator{\lcm}{lcm}

\begin{document}

\title{A Grouping-based Scheduler for Efficient Channel Utilization under Age of Information Constraints}
\author{\IEEEauthorblockN{Lehan~Wang, Jingzhou~Sun, Yuxuan~Sun, Sheng~Zhou, Zhisheng~Niu}
	
	\IEEEauthorblockA{Beijing National Research Center for Information Science and Technology\\
		Department of Electronic Engineering, Tsinghua University, Beijing 100084, P.R. China\\
		\{wang-lh19, sunjz18\}@mails.tsinghua.edu.cn, \{sunyuxuan, sheng.zhou, niuzhs\}@tsinghua.edu.cn}}

\maketitle

\begin{abstract}
We consider a status information updating system where a fusion center collects the status information from a large number of sources and each of them has its own age of information (AoI) constraints. A novel grouping-based scheduler is proposed to solve this complex large-scale problem by dividing the sources into different scheduling groups. The problem is then transformed into deriving the optimal grouping scheme. A two-step grouping algorithm (TGA) is proposed: 1) Given AoI constraints, we first identify the sources with harmonic AoI constraints, then design a fast grouping method and an optimal scheduler for these sources. Under harmonic AoI constraints, each constraint is divisible by the smallest one and the sum of reciprocals of the constraints with the same value is divisible by the reciprocal of the smallest one. 2) For the other sources without such a special property, we pack the sources which can be scheduled together with minimum update rates into the same group. Simulations show the channel usage of the proposed TGA is significantly reduced as compared to a recent work and is $0.42\%$ larger than a derived lower bound when the number of sources is large.
\end{abstract}

\section{Introduction}
Future 5G and 6G networks are expected to energize industrial Internet of Things (IoT), in which smart machines such as moving vehicles operate autonomously by collecting status information from remote sensors. To guarantee prompt and reliable decision-making,  the timeliness of the status information is crucial. For this purpose, a new metric called age of information (AoI) has been proposed in \cite{kaul2012real}. It is the time elapsed since the generation of the latest status information received by the destination. 

AoI has been widely studied by researchers \cite{abdel2019optimized,li2022scheduling,liu2021aion,tang2020minimizing,sun2022status,sun2021age,devassy2018delay,kadota2018scheduling,li2021scheduling,maatouk2022timely,peng2020age,song2020optimal,sun2018information,wang2021uoi}. Most existing works focus on the average AoI over time and different information sources. Since extreme values of AoI are also important, a policy with optimal average age may not be enough for some time-critical applications, such as vehicle-to-vehicle communications \cite{abdel2019optimized}. Therefore, in this paper, the schedulers with AoI constraints are studied. We consider a scenario where a fusion center collects status information from multiple sources. The age of the status information from each source should not exceed their AoI constraints. Assuming at most one source can transmit information in a time slot on a channel, we are interested in the minimum number of channels that can guarantee the AoI constraints of all sources and how to construct a scheduler for them. 

The most relevant works are \cite{li2022scheduling,liu2021aion}. In \cite{li2022scheduling}, the authors assume there is only one channel and can always find a scheduler when $\sum_{n=1}^{N}\frac{1}{d_n}$, the load of AoI constraints of $N$ sources $[d_1, d_2, \ldots, d_N]$, is less than or equal to $\ln2$. In \cite{liu2021aion}, the authors propose an algorithm named Aion, which can construct cyclic schedulers for any AoI constraints. The ratio of the number of channels achieved by Aion to the optimal solution is upper bounded by a function of the update frequency of the sources and the load of AoI constraints. Aion is optimal only when the AoI constraints are consecutively divisible, where $d_{n+1}$ is divisible by $d_n$ for $n\in\{1,\ldots,N-1\}$. For the generalized AoI constraints, Aion narrows the search space to the cyclic schedulers with consecutively divisible average transmission intervals and may not be optimal. In this paper, our grouping method 
efficiently decomposes this complex channel minimization problem into several sub-problems, and does not require the schedulers to provide special average transmission intervals for all the sources. 

The contributions of this paper are summarized as follows:

\begin{itemize}
\item We propose a grouping method to split the problem into multiple small-scale sub-problems by dividing the sources into different groups, and constructing schedulers for each group separately. Based on the grouping method, we then derive an upper bound of the number of channels and transform the channel minimization problem into finding the optimal grouping method.
\item A novel two-step grouping algorithm (TGA) is proposed to solve the grouping problem. Given AoI constraints, we first identify the sources with harmonic AoI constraints, because we can design a fast grouping method and an optimal scheduler for such sources. Under harmonic AoI constraints, all the constraints are divisible by the smallest one, and the sum of reciprocals of the constraints with the same value is divisible by the reciprocal of the smallest AoI constraints. Secondly, for other sources, a heuristic grouping algorithm is proposed. The sources are assigned to the group with the minimum distance. The distance between source $n$ and group $g$ is the difference between $\frac{1}{d_n}$ and the minimum achievable update frequency when $n$ is packed into $g$, where $d_n$ is AoI constraint of $n$. Simulations show the number of channels achieved by the proposed TGA is reduced by 20\% than that of Aion when there are 300 sources. 
\end{itemize}

The rest of the paper is organized as follows. In Section II,
the system model is described.  
The problem formulation and motivations of our grouping method are presented in Section III. The schedulers for harmonic AoI constraints are discussed in Section IV while TGA is proposed in Section V. Finally, the simulation results are shown in Section VI, and the conclusions
are drawn in Section VII.

\section{System Model}

We consider a time-slotted system with one fusion center and $N$ information sources sharing $K$ reliable wireless channels, as shown in Fig.~\ref{system}. At the beginning of each time slot, the scheduled sources transmit the current status to the fusion center. The status information will be received by the center at the end of the slot.  Denote the scheduling decision of source $n$ by $U_n(t)\in\{0,1\}$. If $U_n(t)=1$, then source $n$ is scheduled at time slot $t$ to transmit its status information to the center.

\vspace{-0.2in}

\begin{figure}[ht]
	
	\centering
	\includegraphics[scale=0.2]{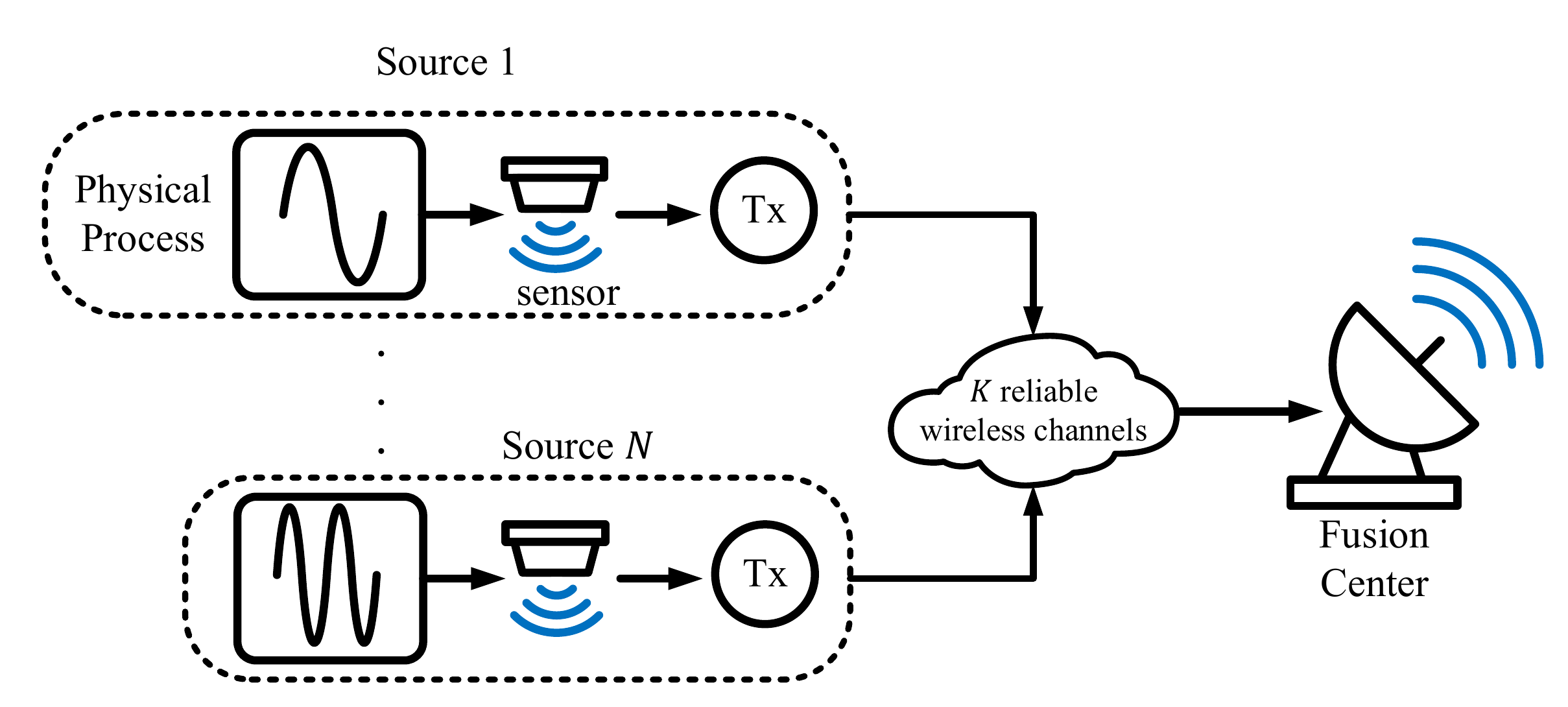}
	\vspace{-0.1in}
	\caption{System model.}
	\label{system}
\end{figure}
	\vspace{-0.1in}

AoI is used to measure the freshness of status information. Denote the generation time of the latest received update from source $n$ and the AoI of source $n$ at time slot $t$ by $G_n(t)$ and $A_n(t)$, respectively. Then $A_n(t)$ is the time elapsed since $G_n(t)$, namely $A_n(t)=t-G_n(t)$. If source $n$ is scheduled at time $t$, then $A_n(t+1)$ will decrease to 1. If $U_n(t)=0$, $A_n(t+1)$ will increase by 1, i.e., 

\begin{equation}
	\begin{split}
		A_n(t+1)=&  \left\{\begin{array}{lc}
			1, &U_n(t)=1;\\
			A_n(t)+1, &U_n(t)=0.\\	
		\end{array}\right.
	\end{split}
\end{equation}

  The AoI constraint of source $n$ is denoted by $d_n\in\mathbb{Z}^{+}$, where $\mathbb{Z}^{+}$ denotes positive integers. Without loss of generality, we assume $d_1\leq d_2 \leq d_3 \cdots\leq d_N$. In addition, each channel can schedule at most one source in a time slot, therefore $K\geq\max_{t}\sum_{n=1}^{N}U_n(t)$. Our goal is to find a scheduler $\bm{\pi}$ which can satisfy the AoI constraints and minimize the number of channels, namely $K$. The $t$-th element of scheduler $\bm{\pi}$ contains scheduling decisions of $N$ sources in time slot $t$, namely $\bm{\pi}(t)=[U_1(t),\cdots,U_n(t), \cdots, U_N(t)]$. 

\section{Problem Formulation and Analysis}

The problem can be formulated as follow: 
 \begin{equation}\label{problem}
	\begin{split}
		\min_{\bm{\pi}} \,\,& K,\\
		s.t.\,\,& A_n(t)\leq d_n,\forall n=1,\cdots,N, \forall t=1,2,\cdots.\\
	\end{split}
\end{equation}

Our objective is to find a scheduler $\bm{\pi}$ which can satisfy the AoI constraints and minimize the number of channels required. The following lemma, also shown in \cite{liu2021aion}, presents the lower bound of this scheduling problem. 

\newtheorem{lemma}{Lemma}
\begin{lemma} \label{lbub}
	Given AoI constraints $\bm{d} =
	[d_1, d_2, \cdots ,d_N]$, we have 
	\begin{equation}
		K^{\star}\geq\left\lceil\sum_{n=1}^{N}\frac{1}{d_n}\right\rceil,
	\end{equation}
	where $K^{\star}$ is the optimal number of channels. 
\end{lemma}

\begin{proof}
	Each source requires at least $\frac{1}{d_n}$ of the time slots of a channel, therefore $K^{\star}\geq\lceil\sum_{n=1}^{N}\frac{1}{d_n}\rceil$.
\end{proof} 

Since scheduler $\bm{\pi}$ possesses infinite elements, solving \eqref{problem} is challenging. If we can narrow the candidate scheduler spaces to cyclic schedulers, which are formally defined in Definition \ref{cyclic scheduler}, the problem can be significantly simplified.  

\newtheorem{definition}{Definition}
\begin{definition} \label{cyclic scheduler}
	A scheduler is cyclic if and only if there exists $C\in\mathbb{Z}^{+}$ such that $U_n(t)=U_n(t+C), \forall n\in\{1,2,\cdots,N\},\forall t$. The cycle length of the scheduler is the smallest integer $C$ which satisfies $U_n(t)=U_n(t+C)$ for $N$ sources.
\end{definition}
The following lemma provides us with the basis to consider only cyclic schedulers and not non-cyclic schedulers.

\begin{lemma} \label{cyclic}
	For any AoI constraints $\bm{d} =
	[d_1, d_2, \cdots ,d_N]$, there exists an optimal cyclic scheduler which only requires $K^{\star}$ channels.
\end{lemma}
\begin{proof}
To prove this lemma, we define the state of $N$ sources at time $t$ as $\bm{S}(t)=[A_1(t),\cdots,A_N(t)]$. With scheduling decision $\bm{\pi}(t)$, the state will transit into $\bm{S}(t+1)=(\bm{J}_{1\times N}-\bm{\pi}(t))\odot\bm{S}(t)+\bm{J}_{1\times N}$, where $\bm{J}_{1\times N}$ is a $N$-dimensional vector whose elements are all $1$, $\odot$ represents element-wise product between two matrices. For any scheduler which can satisfy the AoI constraints, $A_n(t)\leq d_n$ holds for all the sources. Therefore, there are $\prod_{n=1}^{N}d_n$ different states. We assume there exists an optimal scheduler $\bm{\pi}^{\star}$ which requires $K^{\star}$ channels. With $\bm{\pi}^{\star}$, there must be a state which occurs more than once in the time interval from the first slot to the $\left(\prod_{n=1}^{N}d_n+1\right)$-th slot. Let $t_1$ and $t_2$ denote the time slot when the state appears for the first time, and the second time, respectively. Then we can construct an optimal cyclic scheduler $\bm{\pi}^{\prime}$ with cycle length $C=t_2-t_1$. For time slot $t^{\prime}=t+mC$, $t\in[1,C],m\in\mathbb{N}$, let $\bm{\pi}^{\prime}(t^{\prime})=\bm{\pi}^{\star}(t)$, where $\mathbb{N}$ denotes natural numbers. Then $\bm{\pi}^{\prime}$ also satisfies the AoI constraints and requires $K^\star$ channels.   
\end{proof}

Therefore, to solve problem \eqref{problem}, we first determine all the possible transitions among the states and then identify the loop (a transition path starts and ends at the same state) requiring the minimum number of channels. Due to the large state space, searching for the optimal scheduler is computationally prohibitive. A simple idea is to divide $N$ sources into different groups and design schedulers for each group independently. Such an operation can split the large-scale problem into multiple sub-problems and efficiently narrows down the state space.

Based on this idea, we focus on the design of the grouping scheme and the corresponding cyclic scheduler in this paper. 
To discuss what kind of sources can be divided into a same group, we start from the simplest scenario where $d_1=\cdots=d_N$. In this case, a scheduler can be constructed by simply assigning every $d_1$ time slots to each source. Under this cyclic scheduler, the average transmission interval of source $n$ is $l_n=\frac{C}{\sum_{t=1}^{C}U_n(t)}$ and  $l_n=d_1$ holds for each source. Since each channel can schedule at most $d_1$ sources, such a scheduler requires $\lceil N/d_1 \rceil$ channels, which meets the lower bound of the problem. Based on this inspiring observation, we can construct a scheduler named grouping for distinct values (GD) given any AoI constraints $\bm{d}$.  We define the number of distinct values in $\bm{d}$ as $v$. $\bm{u}=\{u_1,\cdots,u_v\}$ includes all the distinct values. Let $o_j$ denote the number of occurrences of $u_j$ in $\bm{d}$. In GD, the sources with the same AoI constraints are packed into a same group. 
Then we design schedulers for each group independently. Since each channel can schedule at most $u_j$ sources with AoI constraints $u_j$, then $o_j$ sources can be scheduled by $\lceil\frac{o_j}{u_j}\rceil$ channels. By assigning different channels to each group, GD requires $\sum_{j=1}^{v}\lceil\frac{o_j}{u_j}\rceil$ channels. Note that when $v=1$ or $\frac{o_j}{u_j}\in\mathbb{Z}^{+}$, $\forall j$, GD is optimal since $\sum_{j=1}^{v}\lceil\frac{o_j}{u_j}\rceil=\lceil\sum_{n=1}^{N}\frac{1}{d_n}\rceil$ in these cases. Based on GD, we can present an upper bound of the problem.

\begin{lemma} \label{ub}
	Given AoI constraints $\bm{d} =
	[d_1, d_2, \cdots ,d_N]$, we have 
	\begin{equation}
		K^{\star}\leq\sum_{j=1}^{v}\left\lceil\frac{o_j}{u_j}\right\rceil.
	\end{equation}
\end{lemma}

\begin{proof}
	The proof is based on the construction of GD.
\end{proof}

Although GD provides the simplest grouping principle and the corresponding cyclic schedulers, GD fails to use the channels efficiently. A channel is \emph{fully utilized} only when all the time slots of the channel are occupied. In $\lceil\frac{o_j}{u_j}\rceil$ channels assigned to sources with AoI constraint $u_j$, there is at most one channel not fully utilized. Therefore, given any AoI constraints $\bm{d}$, there are $v$ channels that are not effectively utilized in the worst case. It is conceivable that when $v$ is large and $o_j$ is relatively small for all the distinct values, the performance of GD is poor. To solve this problem, we should find more effective schedulers to extend the conditions of the sources which can be scheduled in a same group. 

\newtheorem{thm}{Theorem}

\section{Schedulers for Harmonic AoI Constraints}

Before introducing our two-step grouping method, we first consider some special AoI constraints and design the optimal schedulers for them. The theoretical results provided in this section will be the foundation for the first step of our grouping method. We first present a definition:

\begin{definition} \label{semi-harmonic}
	Given $\bm{x}=[x_1,\cdots,x_N]$, $\bm{u}=[u_1,\cdots,u_v]$ includes the distinct values of $\bm{x}$. $\bm{x}$ is harmonic if $x_n\in\mathbb{Z}^{+}, \forall n\in\{1,\cdots,N\}$, and $v\geq2$, $\frac{u_i}{u_1}\in\mathbb{Z}^{+}$,$\forall i\in\{2,\cdots,v\}$. Besides, the number of occurrences of $u_i$ in $\bm{x}$ is an integral multiple of $\frac{u_i}{u_1}$. 
\end{definition}

For example, $\bm{x}=[2,4,4,4,4,6,6,6]$ is harmonic, since both $4$ and $6$ are divisible by the first element in $\bm{x}$ and the occurrences of $4$ and $6$ satisfy the conditions in Definition \ref{semi-harmonic}.  

\begin{thm}\label{semi-harmonic scheduling}
	For any harmonic AoI constraints $\bm{d}=[d_1,\cdots,d_N]$, there exists an optimal cyclic scheduler.
\end{thm}

\begin{proof}
We prove Theorem \ref{semi-harmonic scheduling} by constructing a cyclic scheduler named harmonic scheduler (HS), which requires $\lceil\sum_{n=1}^{N}\frac{1}{d_n}\rceil$ channels under harmonic AoI constraints $\bm{d}$. 
We first construct a scheduler by GD for AoI constraints $\bm{d}^{\prime}$ which contains $\sum_{n=1}^{v}\frac{o_nu_1}{u_n}$ elements with value $u_1$. $o_n$ is the number of occurrences of $u_n$ in $\bm{d}$. Since elements in $\bm{d}^{\prime}$ are the same, GD is optimal for $\bm{d}^{\prime}$. For ease of notation, resource block $(t,k)$ denotes the $t$-th time slot of channel $k$. Define the sequence of resource blocks assigned to the $i$-th element in $\bm{d}^{\prime}$ as $\mathcal{R}_i,i\in\{1,\cdots,\sum_{n=1}^{v}\frac{o_nu_1}{u_n}\}$. The time gap between two consecutive resource blocks in $\mathcal{R}_i$ is $u_1$. Therefore we assign $\{\mathcal{R}_1,\cdots,\mathcal{R}_{o_1}\}$ to $o_1$ sources with AoI constraints $u_1$. For sources with AoI constraints $u_n$, we find the first sequence of resource blocks that has not been fully occupied, denoted by $\mathcal{R}_q$. Then we identify $(t_p,k_p)$, the first available resource block in $\mathcal{R}_q$, and assign resource blocks $(t_p+mu_n,k_p), m\in\mathbb{N}$ to the source. 
Therefore each sequence of resource blocks can schedule at most $\frac{u_n}{u_1}$ sources with AoI constraint $u_n$, and $o_n$ sources with AoI constraints $u_n$ require $\frac{o_nu_1}{u_n}$ sequences of resource blocks. Such requirements can be satisfied because GD provides $\sum_{n=1}^{v}\frac{o_nu_1}{u_n}$ sequences of resource blocks. Since the lower bound of channels are the same for $\bm{d}$ and $\bm{d}^{\prime}$ and HS requires the same number of channels as GD,  HS is optimal given any harmonic AoI constraints.
\end{proof}

Theorem 1 expands the range of sources which can be packed into a same group to the sources with harmonic AoI constraints. Also, HS is more efficient compared with GD. For example, given harmonic $\bm{d}=[2,4,4,4,4,6,6,6]$, GD requires 3 channels while HS requires 2 channels.

Furthermore, we can also handle the cases when $\bm{d}$ is a combination of two harmonic AoI constraints with different bases. The base is defined as the smallest element of the harmonic AoI constraints. 

\begin{lemma} \label{twovalue}
	For AoI constraints with two distinct values: $u_1$ and $u_2$, there exists an optimal cyclic scheduler if 
	$\frac{o_1}{u_1}+\frac{o_2}{u_2}$ is an integer. $o_1,o_2$ is the number of occurrences of $u_1$ and $u_2$ in AoI constraints, respectively. 
\end{lemma}

\begin{proof} 
	We prove Lemma \ref{twovalue} by constructing a scheduler for two distinct values (STV). 
	According to Lemma \ref{lbub}, the lower bound of the minimum number of channels is $K_{\text{LB}}=\frac{o_1}{u_1}+\frac{o_2}{u_2}$. If STV requires $K_{\text{LB}}$ channels for any $u_1,u_2,o_1,o_2$ satisfying the conditions in Lemma \ref{twovalue}, then the proof is completed. 
		\begin{algorithm}[h]\label{STD}
		\caption{Scheduler for two distinct values (STV)}
		\begin{algorithmic}[1]
			\REQUIRE  
			AoI constraint $\bm{d}$ with two distinct values: $\{u_1,u_2\}$. $o_1,o_2$ are the number of occurrences of $u_1,u_2$. 
			\ENSURE 
			A cyclic scheduler $\bm{\pi}$;
			\STATE Set the cycle length of $\bm{\pi}$ as $C=\lcm\{u_1,u_2\}$ and the number of channels required by $\bm{\pi}$ is $K_{\text{LB}}=\frac{o_1}{u_1}+\frac{o_2}{u_2}$. 
			\STATE Divide $C$ slots of $\bm{\pi}$ into $C/\gcd(u_1,u_{2})$ groups, each with $\gcd(u_{1},u_{2})K_{\text{LB}}$ resource blocks. 
			\STATE Assign $\frac{o_1\gcd(u_{1},u_{2})}{u_{1}}$ resource blocks in the first group to sources with AoI constraint $u_1$. In the other groups, assign resource blocks in the same positions as in the first group to sources with AoI constraint $u_1$.
			\STATE Schedule sources with AoI constraint $u_2$ similarly. 
		\end{algorithmic}
	\end{algorithm}
	
	For generality, we assume $u_1<u_2$ and $K_{\text{LB}}=b$. 
	Define $\gcd(u_1, u_2)$ as the greatest common divisor of $u_1$ and $u_2$. Then according to linear Diophantine equation, $o_1$ and $o_2$ must satisfy the following conditions: $o_1=x_0+mu_1/\gcd(u_1, u_2)$, $o_2=y_0-mu_2/\gcd(u_1, u_2)$, where $m,x_0,y_0$ are integers. $(x_0, y_0)$ is any set of solutions to equation $\frac{x}{u_1}+\frac{y}{u_2}=b$. 
	If we choose $x_0=0, y_0=bu_2$, then we can prove that $o_1$ is divisible by $u_1/\gcd(u_1, u_2)$, and  $o_2$ is divisible by $u_2/\gcd(u_1, u_2)$.

	If $\gcd(u_1,u_2)=1$, then $o_1, o_2$ are multiples of $u_1, u_2$, respectively. In this case, STV can be reduced to GD. When $\gcd(u_1,u_2)>1$, we set the cycle length of STV as the least common multiple of $u_1$ and $u_2$, denoted by $\lcm(u_1,u_2)$. Then we divide $\lcm(u_1,u_2)$ time slots in a cycle into $a=\lcm(u_1,u_2)/\gcd(u_1, u_2)$ groups. 
	Each group includes $\gcd(u_1, u_2)$ consecutive time slots 
	 and $b\gcd(u_1, u_2)$ resource blocks. 
	
	We then separate $o_1$ sources into the first $u_1/\gcd(u_1, u_2)$ groups, each group contains $\frac{o_1\gcd(u_1, u_2)}{u_1}$ sources and each source is assigned one resource block. Besides, positions of the resource blocks assigned to the sources with the same AoI constraint should be the same in each group. 
	For the $j$-th group ($j=eu_1/\gcd(u_1, u_2)+a^{\prime}$), the scheduler just `copies' the scheduling decisions in group $a^{\prime}$, where $a^{\prime}\in\{1,\cdots,u_1/\gcd(u_1, u_2)\},e\in\{1,\cdots,\lcm\{u_1,u_2\}/u_1-1\}$. The sources with AoI constraints $u_2$ can be scheduled similarly.
	
	Since $\frac{o_1\gcd(u_1, u_2)}{u_1}+\frac{o_2\gcd(u_1, u_2)}{u_2}=b\gcd(u_1,u_2)$, STV requires $K_{\text{LB}}$ channels. The lemma is proved.
\end{proof}

\begin{thm} \label{exclude1}
	For two harmonic AoI constraints $\bm{d}_1 =
	[d_{1,1}, \cdots ,d_{1,i}]$ with length $i$ and $\bm{d}_2 =
	[d_{2,1}, \cdots ,d_{2,j}]$ with length $j$, there exists an optimal cyclic scheduler for the $i+j$ sources if $\sum_{n=1}^{i}\frac{1}{d_{1,n}}+\sum_{n=1}^{j}\frac{1}{d_{2,n}}\in\mathbb{Z}^{+}$. 
\end{thm}

\begin{proof}
	The proof is based on construction. We first use STV to construct a cyclic scheduler for $n_1$ sources with AoI constraint $d_{1,1}$ and $n_2$ sources with AoI constraint $d_{2,1}$, where $n_1=d_{1,1}\sum_{n=1}^{i}\frac{1}{d_{1,n}}$, $n_2=d_{2,1}\sum_{n=1}^{j}\frac{1}{d_{2,n}}$. Since $\frac{n_1}{d_{1,1}}+\frac{n_2}{d_{2,1}}$ is an integer, STV is optimal according to Lemma \ref{twovalue}. Let $\mathcal{R}_1$ and $\mathcal{R}_2$ denote the sequences of resources blocks allocated to the $n_1$ and $n_2$ sources. We use HS to distribute $\mathcal{R}_1$ and $\mathcal{R}_2$ to the sources with AoI constraints $\bm{d}_1$ and $\bm{d}_2$. Therefore, such a scheduler requires $\frac{n_1}{d_{1,1}}+\frac{n_2}{d_{2,1}}$ channels, which is also the lower bound of the number of channels. The proof is completed.
\end{proof}

\section{Two-step Grouping Algorithm}

Now we can design optimal schedulers for harmonic AoI constraints. 
However, how to construct a scheduler with efficient channel utilization for generalized AoI constraints is still a challenging problem.  

To utilize the theoretical results obtained in section IV, we propose a two-step grouping algorithm (TGA): 
\begin{itemize}
\item Step 1: Given a set of sources $\bm{N}=\{1,\cdots,N\}$ and the corresponding AoI constraints $\bm{d}=[d_1,\cdots,d_N]$, we first search the sources with harmonic AoI constraints and design a scheduler for them.
\item Step 2: For the other sources without harmonic AoI constraints, we design a new scheduler with optimized channel utilization for them.
\end{itemize}

In the following of this section, we first introduce our algorithm of identifying the sources with harmonic AoI constraints and constructing schedulers for them. Secondly, we display the grouping algorithm for the other sources without such a special property.
 
\subsection{Step 1: Harmonic Sources Identifying \& Scheduling}

According to Theorems \ref{semi-harmonic scheduling} and \ref{exclude1}, if $\bm{d}$ is harmonic or $\bm{d}$ includes any two harmonic AoI constraints satisfying the conditions in Theorem \ref{exclude1}, the corresponding sources can also be packed in a group for efficient channel utilization.

Therefore, if $\bm{d}$ is a combination of $h$ harmonic AoI constraints, denoted by $\bm{d}_1,\cdots,\bm{d}_h$, we can divide the sources into at most $h$ groups. For each group, we use STV and HS to design a scheduler and assign channels independently. The number of channels required by such a procedure is no more than $\sum_{i=1}^{h}Q_i$, where $Q_i=\lceil\sum_{j\in\bm{d}_i}\frac{1}{d_{i,j}}\rceil$ and $\bm{d}_i=[d_{i,1},\cdots,d_{i,j},\cdots]$.

Based on the analysis above, our harmonic sources identifying (HSI) algorithm is proposed in Algorithm 2. Given AoI constraints $\bm{d}$, we first derive $\bm{u}=[u_1,\cdots,u_v]$ which includes all the distinct values of $\bm{d}$, and arrange $\bm{u}$ in ascending order. HSI is divided into two parts. The first part is completed in $v$ iterations. In the $i$-th iteration, we search the sources which possess harmonic AoI constraints with base $u_i$. $\bm{d}_{i}=[d_{i,1},\cdots,d_{i,m_i}]$ with length $m_i$ denotes the harmonic AoI constraints with base $u_i$. Besides, the elements in $\bm{d}_{i}$ is arranged in ascending order. If we directly construct a scheduler for $\bm{d}_i$ by HS, there will be at most one channel which is not fully utilized. This is due to the introduction of GD in HS. Since the channels are assigned to each group independently, such an inefficient use of channel seriously degrade the performance of the grouping. Therefore, we only construct a scheduler based on HS for the first $p$ elements in $\bm{d}_{i}$, where $p$ satisfies $\sum_{n=1}^{p}\frac{1}{d_{i,n}}=\lfloor \sum_{n=1}^{m_i}\frac{1}{d_{i,n}}\rfloor$. The other $m_i-p$ sources will be scheduled in the second part of HSI or by the heuristic grouping algorithm, which is discussed in the following subsection. After the $i$-th iteration, $p$ sources with AoI constraints $[d_{i,1},\cdots,d_{i,p}]$ are excluded from $\bm{N}$.

\begin{algorithm}[h]\label{pre}
	\caption{Harmonic Sources Identifying (HSI)}
	\begin{algorithmic}[1]
		\REQUIRE  
		AoI constraint vector $\bm{d}$ and the set of sources $\bm{N}$;
		\ENSURE 
		Sources with harmonic AoI constraints $\bm{N}^{\prime}$, other sources with generalized constraints $\bar{\bm{N}}^{\prime}=\bm{N}\setminus\bm{N}^{\prime}$, a cyclic scheduler $\bm{\pi}_{H}$ for $\bm{N}^{\prime}$.
		\STATE Initialize $\bm{N}^{\prime}=\emptyset$. $\bm{u}=\{u_1,\cdots,u_v\}$ includes distinct values in $\bm{d}$. Arrange $\bm{u}$ in ascending order. $o_j$, the $j$-th element in $\bm{o}$, is the number of occurrences of $u_j$ in $\bm{d}$. 
		\FOR{$i\in[1,v]$}
		\STATE 
		$\bm{f}_i=\{j|\frac{u_j}{u_i}\in\mathbb{Z}^{+}\}$, $\bm{d}_i=\emptyset$
		\FOR{$n\in\bm{f}_i$}
		\STATE Put $\lfloor\frac{o_{n}u_i}{u_n}\rfloor \frac{u_{n}}{u_i}$ elements with value of $u_n$ into $\bm{d}_i$.
		\ENDFOR
		\STATE 
		Arrange 
		$\bm{d}_i=[d_{i,1},\cdots,d_{i,m_i}]$ in ascending order. 
		\STATE Find $p$ such that $\sum_{n=1}^{p}\frac{1}{d_{i,n}}=\lfloor \sum_{n=1}^{m_i}\frac{1}{d_{i,n}}\rfloor$. 
		Add $p$ sources with AoI constraints $[d_{i,1},\cdots,d_{i,p}]$ into $\bm{N}^{\prime}$. 
		\STATE Design a scheduler $\bm{\pi}_i$ for the $p$ sources based on HS.
		\STATE $\bar{\bm{N}}^{\prime}=\bm{N}\setminus\bm{N}^{\prime}$. Update $\bm{o}$ for sources in $\bar{\bm{N}}^{\prime}$.   
		\ENDFOR
		\FOR{$i\in[1,v]$}
		\FOR {$j\in \{j|\gcd(u_j,u_i)>1,u_j>u_i, \frac{u_j}{u_i}\notin\mathbb{Z}^{+}\}$}
		\STATE Repeat Line (3)-(7) and obtain $\bm{d}_i$ with length $m_i$.
		\STATE 
		$\bm{f}_j=\{e|\frac{u_e}{u_j}\in\mathbb{Z}^{+},\frac{u_e}{u_i}\notin\mathbb{Z}^{+}\}$
		\STATE Repeat Line (4)-(7) and obtain $\bm{d}_j$ with length $m_j$. 
		\STATE Let $s_i=d_{i,1}\sum_{n=1}^{m_i}\frac{1}{d_{i,n}}$,$s_j=d_{j,1}\sum_{n=1}^{m_j}\frac{1}{d_{j,n}}$
		\STATE Let $b=\lfloor \frac{s_i}{u_i}+\frac{s_j}{u_j}\rfloor$, find the largest positive integer $s_i^{\prime}\leq s_i$ such that $s_j^{\prime}=\frac{1}{u_i}(bu_iu_j-u_js_i^{\prime})\in\mathbb{Z}^{+}$
		\IF {$s_i^{\prime}$ exists and $ s_j^{\prime} \leq s_j$}
		\STATE 
		Find $p_i,p_j$ such that $\sum_{n=1}^{p_z}\frac{1}{d_{z,n}}=\frac{s_z^{\prime}}{u_z}$ , $z\in\{i,j\}$. Add $p_i+p_j$ sources with AoI constraints $[d_{i,1},\cdots,d_{i,p_i},d_{j,1},\cdots,d_{j,p_j}]$ into $\bm{N}^{\prime}$. 
		\STATE Design a scheduler $\bm{\pi}_i^{\prime}$ for the $p_i+p_j$ sources based on STV and HS.
		\STATE $\bar{\bm{N}}^{\prime}=\bm{N}\setminus\bm{N}^{\prime}$ and update $\bm{o}$ for sources in $\bar{\bm{N}}^{\prime}$.   
		\ENDIF
		\ENDFOR
		\ENDFOR
		\STATE $\bm{\pi}_H$ is a combination of $\bm{\pi}_i$ and $\bm{\pi}_i^{\prime}$, $\forall i\in\{1,\cdots,v\}$. The number of channels required by $\bm{\pi}_H$ is $\sum_{n\in\bm{N}^{\prime}}\frac{1}{d_n}$.
	\end{algorithmic}
\end{algorithm}	

In the second part, for distinct value $u_i$, we find sources with AoI constraint $u_j$ such that $\gcd(u_i,u_j)>1$. We first construct the harmonic AoI constraints $\bm{d}_i, \bm{d}_j$ with bases of  $u_i$ and $u_j$.  Then we find the sources with $\bm{d}_i, \bm{d}_j$ satisfying the conditions in Theorem \ref{exclude1} and design an optimal scheduler for them based on STV and HS. After the two parts, an optimal scheduler is obtained for all the sources searched by HSI. In addition, we do not take the cases when $\gcd(u_i,u_j)=1$ into account in the second part. This is because the sum of reciprocals of constraints in the two harmonic vectors will be integers if $\gcd(u_i,u_j)=1$, and such sources are found in the first part.

\subsection{Step 2: heuristic grouping algorithm}

After HSI, we find harmonic groups and design an optimal cyclic scheduler for them. For the other sources without harmonic AoI constraints, HS is not optimal, hence we have to design a new scheduler for them. Before introducing the heuristic grouping algorithm (HGA), a definition is presented. 

\begin{algorithm}[h]\label{SS}
	\caption{Consecutively divisible Scheduler (CS)}
	\begin{algorithmic}[1]
		\REQUIRE  
		 $\bm{l}=[l_1,\cdots,l_N]$ with $l_1\leq\cdots\leq l_N$;
		\ENSURE 
		A cyclic scheduler $\bm{\pi}_S$;
		\STATE Find the minimum integer $a$ satisfying $al_1\in\mathbb{Z}^{+}$. Construct $\bm{\pi}^{\prime}$ requiring $\lceil\frac{1}{a}\lceil\sum_{n=1}^{N}\frac{1}{l_n}\rceil\rceil$ channels. The cycle length $C_{\bm{\pi}^{\prime}}=al_N$ if $l_N\in\mathbb{Z}^+$, otherwise $C_{\bm{\pi}^{\prime}}=a^2l_N$.
		\STATE Let $\bm{M}=[M_1,\cdots,M_{C_{\bm{\pi}^{\prime}}/a}]$ with each $M_j=\lceil\sum_{n=1}^{N}\frac{1}{l_n}\rceil$
		\FOR {$i\in[1,N]$}
		\STATE Find $p$ such that $M_p$ is the first largest element in $[M_1,\cdots,M_{\lceil l_i\rceil}]$
		\STATE Find available  resource block $(t_p,k_p)$ such that  $\sum_{n=1}^{N}U_n(t_p)$ is the first smallest element in $[\sum_{n=1}^{N}U_n((p-1)a+1),\cdots,\sum_{n=1}^{N}U_n(\min\{pa,al_i\})]$.
		\STATE Assign resource blocks $(t_p+mal_i,k_p)$, $m\in\{0,\cdots,C_{\bm{\pi}^{\prime}}/(al_i)-1\}$ to the source $i$. Let $M_r=M_r-1$, for $r=\lceil\frac{t_p+mal_i}{a}\rceil$, $m\in\{0,\cdots,C_{\bm{\pi}^{\prime}}/(al_i)-1\}$
		\ENDFOR
		\STATE Set the cycle length of $\bm{\pi}_S$ as $C_S=C_{\bm{\pi}^{\prime}}/a$. 
		\STATE The number of channels required by $\bm{\pi}_S$ is $\lceil\sum_{n=1}^{N}\frac{1}{l_n}\rceil$. 
		\FOR {$i\in[1,C_S]$}
		\STATE $\bm{N}_i$ denotes the sources scheduled by $\bm{\pi}^{\prime}$ in slots $[(i-1)a+1,ia]$. 
		\STATE In $\bm{\pi}_S$, set $U_n(i)=1$ for each source $n$ in $\bm{N}_i$.                                                                         
		\ENDFOR
		
	\end{algorithmic}
\end{algorithm}

\begin{definition} \label{harmonic}
	A vector $\bm{x}=[x_1,x_2,\cdots,x_N]$ is consecutively divisible if $x_i\in \mathbb{R}^{+}$, $x_i\geq1$, $\forall i\in\{1,\cdots,N\}$ and $x_i/x_{i-1}\in\mathbb{Z}^{+},\forall i\in\{2,\cdots,N\}$. 
\end{definition}

In Definition \ref{harmonic}, $\mathbb{R}^{+}$ denotes positive real numbers. If $[d_1,\cdots,d_N]$ is consecutively divisible, then a consecutively divisible AoI constraints scheduler (CAS) requiring $\lceil\sum_{n=1}^{N}\frac{1}{d_n}\rceil$ channels can be constructed \cite{liu2021aion}. 
CAS can be constructed within $N$ iterations. 
In the $n$-th iteration, we find the smallest $t_n\leq d_n$ and the smallest $k_n\leq K$ such that resource block $(t_n,k_n)$ is available. Then resource blocks $\mathcal{R}_n=\{(t,k_n)|t=t_n+md_n, m\in\mathbb{N}\}$ are assigned to source $n$, namely $U_n(t_n+md_n)=1$. According to Lemma \ref{lbub}, CAS is optimal for consecutively divisible AoI constraints. 

For generalized AoI constraints, CAS can not be used directly. According to \cite{liu2021aion}, if there exist consecutively divisible average transmission intervals $\bm{l}=[l_1,l_2,\cdots,l_N]$ such that $l_n\leq d_n, \forall n\in\{1,\cdots,N\}$, a consecutively divisible scheduler (CS) requiring $\lceil\sum_{n=1}^{N}\frac{1}{l_n}\rceil$ channels can be constructed as a generalized version of CAS. The details of CS is shown in Algorithm 3.

With CS, we can design a scheduler under any AoI constraints. Given generalized AoI constraints, a direct idea is to find consecutively divisible average transmission intervals requiring the minimum number of channels based on CS, i.e.,
	  \begin{equation}\label{Aion}
	\begin{split}
		\min_{\bm{l}} \,\,& \left\lceil\sum_{n=1}^{N}\frac{1}{l_n}\right\rceil,\\
		s.t.& 1\leq\l_n\leq d_n,\forall n=1,2,\cdots,N,\\
		& \bm{l}\text{ is consecutively divisible.}
	\end{split}
\end{equation}

 Let $K_1$ denote the number of channels required by the solution of \eqref{Aion}. However, such a method may lead to inefficient channel utilization, which is shown in the following example.

\begin{Example} \label{example 2}
	Consider sources $\{A, B, C, D, E, F, G, H, I, J\}$ and generalized AoI constraints $\bm{d}=[3,5,5,5,6,6,6,7,7,7]$, then $\bm{l}^{\star}=[2.5,5,5,5,5,5,5,5,5,5]$ is the solution of \eqref{Aion}. The cyclic scheduler constructed by CS is:
	\vspace{-0.06in}
	\begin{center}
		\setlength{\tabcolsep}{1.3mm}{
			\begin{tabular}{l c c c c c c c c c c c p{ 0.4 cm}} 		
				time slot&1&2&3&4&5&6&7&8&9&10\\
				channel 1& A & B & A & C & D & A & B & A & C & D & $\cdots$	\\
				channel 2 &E&F&G & H & I & E & F & G & H & I & $\cdots$\\
				channel 3& J & $\square$ & $\square$ & $\square$ & $\square$ & J & $\square$ & $\square$ & $\square$ & $\square$ & $\cdots$	\\
		\end{tabular}}
	\end{center}
	CS requires 3 channels while the lower bound of the number of channels is 2. Besides, channel 3 is not fully utilized. We can construct a cyclic scheduler which meets the lower bound.
	\begin{center}
		\setlength{\tabcolsep}{1.3mm}{
			\begin{tabular}{l c c c c c c c c c c c c c} 		
				time slot&1&2&3&4&5&6&7&8&9&10&11&12\\
				channel 1& A & B & A & C & D & A & B & A & C & D & A & B & $\cdots$	\\
				channel 2 &E&F&G & H & I &J& E & F & G & H & I & J& $\cdots$\\
				
		\end{tabular}}
	\end{center}
	The scheduler above possesses average transmission intervals $\bm{l}=[2.5,5,5,5,6,6,6,6,6,6]$, which is not consecutively divisible. 
	The cycle length $C$ is the least common multiple of cycle length of channel 1 and 2, namely $C=\lcm\{5,6\}=30$. 
\end{Example}

$\bm{l}$ in Example \ref{example 2} is actually a combination of two different consecutively divisible vectors, namely $\bm{l}_1=[2.5,5,5,5]$ and $\bm{l}_2=[6,6,6,6,6,6]$. The cyclic scheduler is also a combination of two schedulers derived by CS, with average transmission interval vectors $\bm{l}_1$ and $\bm{l}_2$, respectively. This phenomenon leads to the following theorem. 

\begin{thm}\label{thm2}
	For any average transmission interval vector
	$\bm{l}$ which is a combination of $h$ consecutively divisible vectors $\bm{l}_1,\bm{l}_2,\cdots,\bm{l}_h$, there exists a cyclic scheduler requiring $\sum_{j=1}^{h}Q_j$ channels, where $Q_j=\lceil\sum_{i\in \bm{l}_j}\frac{1}{l_{j,i}}\rceil$ and $\bm{l}_j=[l_{j,1},\cdots,l_{j,i},\cdots]$.
\end{thm}

Based on Theorem \ref{thm2}, we can also divide the sources into different groups and design a scheduler for each group independently given generalized AoI constraints. We still need to find a grouping method for sources such that the corresponding average transmission intervals lead to minimum number of channels. The problem can be formulated as follow:
\begin{equation}\label{grouping problem}
	\begin{split}
		\min_{\mathcal{G},\bm{l}} & \sum_{j\in\mathcal{G}}\left\lceil{ \sum_{i\in \bm{g}_j}\frac{1}{l_i}}\right\rceil,
		\\
		s.t.&  \left\{\begin{array}{lc}
			1\leq l_n\leq d_n, \forall n=1,2,\dots,N\\
			\bm{l}_j \text{ is consecutively divisible}, \forall j=1,\dots,|\mathcal{G}|,\\
		\end{array}\right.
	\end{split}
\end{equation}
where $\mathcal{G}$ represents a grouping method and $|\mathcal{G}|$ is the number of groups. $\bm{g}_j$ includes the sources in the $j$-th group. The average transmission interval of source $n$ is $l_n$ and $\bm{l}_j$ denotes the consecutively divisible average transmission intervals of the $j$-th group.

However, searching for the optimal grouping method is computationally prohibitive. At least $2^{N-1}-1$ grouping methods already exist for the case of just dividing $N$ sources into two groups. 
We have three naive tricks to narrow down the search space: 1) Problem \eqref{Aion} is a special case of problem \eqref{grouping problem} when we only have one group, 
then if $K_{1}$ equals the lower bound in Lemma \ref{lbub}, there is no need to find other grouping schemes. 
2) When $K_{1}$ is larger than the lower bound, to achieve better performance, the optional number of groups should be larger than 1. Additionally, by Theorem \ref{thm2}, each group is assigned at least one channel. Therefore the optional number of groups is upper bounded by $K_{1}-1$. 3) The grouping scheme $\mathcal{G}$ with the lower bound $\sum_{j\in\mathcal{G}}\lceil{ \sum_{i\in \bm{g}_j}\frac{1}{d_i}}\rceil> K_{1}$ can be excluded.

Solving problem \eqref{grouping problem} is challenging with the tricks above because there are still many optional grouping methods left. Therefore, we propose a heuristic grouping algorithm (HGA).

Note that the grouping problem in \eqref{grouping problem} is similar to a clustering problem in the field of machine learning. Some definitions inspired by the clustering are presented. 

\begin{definition} \label{center}
	The center of a group is the source whose average transmission interval obtained by solving problem \eqref{Aion} for the sources in this group, equals its AoI constraint. 
\end{definition}

Based on property 2 in \cite{liu2021aion}, if $l^{\star}=\{l_1^{\star},\cdots,l_N^{\star}\}$ is the solution to \eqref{Aion} with minimum $\sum_{n=1}^{N}\frac{1}{l_n^{\star}}$, then there exists at least one element $i$ in $\bm{l}^{\star}$ such that $l_i^{\star}=d_i$. Since the average transmission interval obtained by solving problem \eqref{Aion} is consecutively divisible, this property can be used to narrow down the search space of the grouping method. To achieve better performance, we should make the gap between the AoI constraints and the average transmission intervals obtained by solving problem \eqref{Aion} in each group as small as possible. Given source $\mu_g$ as the center of group $g$, we define the distance between group $g$ and any source $n$ as follows:

\begin{algorithm}[h]
	\caption{Heuristic Grouping Algorithm (HGA)}
	\begin{algorithmic}[1]
		\REQUIRE  
		The set of sources $\bm{N}$ and the AoI constraints $\bm{d}$;
		\ENSURE 
		A cyclic scheduler $\bm{\pi}_G$;
		\STATE Obtain $\bm{u}=\{u_1,\cdots,u_v\}$ including all the distinct values in $\bm{d}$. For sources in $\bm{N}$, solve problem \eqref{Aion} for average transmission intervals $\{l_1,\cdots,l_N\}$ and $K_{1}=\lceil\sum_{i\in\bm{N}}\frac{1}{l_i}\rceil$.
		\FOR{$i\in[2,K_1-1]$}
		\STATE The optional set of centers: $\bm{u}_{ce}=\{\bm{z}|\bm{z}\subset\bm{u},|\bm{z}|=i\}$
		\FOR {each optional set of centers}
		\STATE Put sources in the group with the minimum distance. $\bm{G}$ denotes the groups with unused part larger than $\gamma$. 
		\FOR{each group $j$ in $\bm{G}$}
		\STATE $\bm{r}=\{D_{j,n}+\frac{1}{d_n}|n\in\bm{g}_j\}$, where $\bm{g}_j$ includes all the sources in $j$-th group and has $|\bm{g}_j|$ elements. 
		\STATE Arrange $\bm{r}$ in descending order. Find the maximum $m$ such that $\sum_{n=1}^{m}r_n\leq \lfloor\sum_{n=1}^{|\bm{g}_j|}r_n\rfloor$. 
		\STATE Put the sources corresponding to $\{r_{m+1},\cdots,r_{|\bm{g}_j|}\}$ to other group with the minimum distance and enough unused part. If we can not find such groups, put the sources into the first group.
		\ENDFOR
		\STATE Derive the minimum number of channels for each group by solving problem \eqref{Aion}. $K_G$ denotes the number of channels required by all the groups. 
		\IF {$K_G=\lceil\sum_{n=1}^{N}\frac{1}{d_n}\rceil$}
		\STATE Construct $\bm{\pi}_G$ for the current grouping scheme by CS. Terminate the algorithm.
		\ENDIF
		\ENDFOR
		\ENDFOR
		\STATE Obtain the optimal grouping method over all the optional number of groups and choices of centers. Construct $\bm{\pi}_G$ by CS.		
	\end{algorithmic}
\end{algorithm}

\begin{definition} \label{distance}
	The distance between group $g$ and source $n$ is:
	\begin{equation}
		\begin{split}
			D_{g,n}=&  \left\{\begin{array}{lc}
				\frac{1}{\lfloor\frac{d_n}{d_{\mu_g}}\rfloor d_{\mu_g}}-\frac{1}{d_n}, &d_n\geq d_{\mu_g};\\
				\frac{\lceil \frac{d_{\mu_g}}{d_n} \rceil}{d_{\mu_g}}-\frac{1}{d_n}, &d_n<d_{\mu_g}.\\	
			\end{array}\right.
		\end{split}
	\end{equation}	  
\end{definition}

If $n$ is put into group $g$, then there are two different cases: 1) when $d_n\geq d_{\mu_g}$, $d_{\mu_g}$ should be a divisor of $l_n$, then the largest $l_n$ is $\lfloor\frac{d_n}{d_{\mu_g}}\rfloor d_{\mu_g}$; 2) when $d_n< d_{\mu_g}$, $d_{\mu_g}$ should be divisible by $l_n$, hence $l_n\leq d_{\mu_g}\frac{1}{\lceil \frac{d_{\mu_g}}{d_n} \rceil}$. Therefore, the distance between group $g$ and source $n$ is the difference between the reciprocal of $d_n$ and the largest average transmission interval of source $n$ when $n$ is put into group $g$. In addition, there may be more than one source with the same average transmission interval as AoI constraint in a group. For simplicity, the distance is calculated based on only one center. The simulation results in the next section validate this simplification. 

To approach the minimum number of channels, we should put the sources into the group with the minimum distance. However, such a procedure may lead to inefficient use of channels. In Example \ref{example 2}, if the centers of group 1 and 2 are $A$ and $B$. Then $\{E,F,G,H,I,J\}$ are put into group 1 while $C,D$ are packed into group 2. It is frustrating that such an operation requires 3 channels. If we rearrange source $I,J$ and put them in group 2, then we only need 2 channels. 

To illustrate this phenomenon, we define the utilization of an average transmission interval vector $\bm{l}=[l_1,\cdots,l_N]$ as $\sum_{n=1}^{N}\frac{1}{l_n}$. The number of channels assigned to $\bm{l}$ is  $\lceil\sum_{n=1}^{N}\frac{1}{l_n}\rceil$, when $\bm{l}$ is consecutively divisible. Then we use $\lceil\sum_{n=1}^{N}\frac{1}{l_n}\rceil-\sum_{n=1}^{N}\frac{1}{l_n}$ to approximate the unused part of the channels. A large unused part indicates the inefficient utilization of channels. When $n$ is packed into group $g$, $l_n$ can be approximated by reciprocal of $D_{g,n}+\frac{1}{d_n}$. Therefore, if group 1 contains $\{A,E,F,G,H,I,J\}$, then the approximated average transmission intervals are $[3,6,6,6,6,6,6]$ and the unused part of the channel is $2-\frac{1}{3}-\frac{6}{6}=\frac{2}{3}$. For group 2 with $\{B,C,D\}$, the unused part of the channel is $1-\frac{3}{5}=\frac{2}{5}$. Both two groups inefficiently use the channels. If group 1 and 2 contain $\{A,E,F,G,H\}$ and $\{B,C,D,I,J\}$, the approximated average transmission interval vectors are $[3,6,6,6,6]$ and $[5,5,5,5,5]$ while the unused part of channels are zero for both groups. Therefore, after putting the sources to the group with the minimum distance, we should rearrange the sources which lead to inefficient use of the channels. 

\begin{algorithm}[h]\label{TGA}
	\caption{Two-step Grouping Algorithm (TGA)}
	\begin{algorithmic}[1]
		\REQUIRE  
		The set of sources $\bm{N}$ and the AoI constraints $\bm{d}$;
		\ENSURE 
		A cyclic scheduler $\bm{\pi}_T$;
		\STATE Search harmonic sources by HSI, construct $\bm{\pi}_H$
		\STATE For the other sources in $\bm{N}$, design $\bm{\pi}_G$ by HGA 
		\STATE $\bm{\pi}_T$ is the combination of $\bm{\pi}_H$ and $\bm{\pi}_G$. The number of channels required by $\bm{\pi}_T$ is the sum of numbers of channels required by $\bm{\pi}_H$ and $\bm{\pi}_G$.
	\end{algorithmic}
\end{algorithm}

Finally, we propose HGA in Algorithm 4. We first derive $K_1$ by solving problem \eqref{Aion}, 
hence the optional number of groups is $\{2,\cdots,K_{1}-1\}$. 
Secondly, we assume the centers of each group have distinct AoI constraints. If there are two centers $\mu_g$ and $\mu_{g^{\prime}}$ with the same AoI constraint, then for any source $n$, the two groups are the same because distance $D_{g,n}=D_{g^{\prime},n}$. Therefore, if we divide the sources into $i$ groups, there are $\tbinom{v}{i}$ ways to choose centers, where $v$ is the number of distinct values in $\bm{d}$. For each optional combination of centers, we put the sources into the group with the minimum distance. After that, for each group with large unused part of channels (unused part larger than $\gamma$), we find the sources leading to the inefficient utilization and put them into the other groups with the minimum distance and enough unused parts of channels. Here, `enough' means that if source $n$ is rearranged and put into group $g^{\prime}$, then $D_{g^{\prime},n}+\frac{1}{d_n}$ is no larger than the unused part of $g^{\prime}$. 
When the grouping method for a certain combination of centers is derived, we obtain the number of channels by solving \eqref{Aion} for each group. If the number of channels required by the current grouping scheme equals the lower bound of the problem, we stop searching for a better solution and terminate the algorithm. Otherwise, we will construct a grouping scheme for each optional number of groups and combination of centers to derive the optimal grouping method. Finally, we construct scheduler $\bm{\pi}_G$ by using CS for each group.

\subsection{Complexity}

In HSI, we first search the sources with harmonic AoI constraints with a complexity of $O(N^3)$. Then we construct a scheduler for them with a complexity of $O(C_HN)$, $C_H$ is the cycle length of $\bm{\pi}_H$ constructed by HSI. In HGA, we use a fast algorithm proposed in \cite{liu2021aion} to solve problem \eqref{Aion} and the complexity is $O(N^4(d_N)^2)$. Since we have at most $2^v$ possible combinations of centers, the complexity of HGA is $O(2^vN^4(d_N)^2)$, where $v$ is the number of distinct values in $\bm{d}$. In conclusion, TGA has a complexity of $O(2^vN^4(d_N)^2)$.

\subsection{Combination of HSI and HGA}

As mentioned earlier, TGA is a combination of HSI and HGA, shown in Algorithm 5. Although HGA can design a scheduler for any AoI constraints, HSI plays a key role in TGA. Firstly, the number of sources needed to be handled by HGA is decreased after HSI. Secondly, 
since the optional number of groups is upper bounded by $K_1-1$, HGA can reduce $K_1$ if the sources founded by HSI, denoted by $\bm{N}^{\prime}$, is not an emptyset. Therefore HSI can efficiently narrow down the search space of HGA. Thirdly, since HSI has strong theoretical guarantee, taking HSI as the first step of TGA can achieve better performance compared with pure HGA. Denote the other sources with generalized AoI constraints by $\bar{\bm{N}^{\prime}}=\bm{N}\setminus\bm{N}^{\prime}$. Define $\bm{l}^{\star}=[l_1^{\star},\cdots,l_N^{\star}]$ and $\bm{l}^{\prime}=[l_1^{\prime},\cdots,l_n^{\prime},\cdots]$ as the optimal average transmission intervals of problem \eqref{Aion} for $N$ sources and sources in $\bar{\bm{N}^{\prime}}$, respectively. Let $K_{WA}$ and $K_{OA}$ denote the number of channels achieved by TGA with and without HSI when we only allow the sources in $\bar{\bm{N}^{\prime}}$ to be packed into the same group, then  
\begin{equation}
	\begin{split}
		K_{OA}&=\left\lceil\sum_{n\in\bm{N}}\frac{1}{l_n^{\star}}\right\rceil=\left\lceil\sum_{n\in\bm{N}^{\prime}}\frac{1}{l_n^{\star}}+\sum_{n\in\bar{\bm{N}^{\prime}}}\frac{1}{l_n^{\star}}\right\rceil\\
		&\overset{(a)}{\geq}\left\lceil\sum_{n\in\bm{N}^{\prime}}\frac{1}{d_n}+\sum_{n\in\bar{\bm{N}^{\prime}}}\frac{1}{l_n^{\star}}\right\rceil\overset{(b)}{=}\left\lceil\sum_{n\in\bar{\bm{N}^{\prime}}}\frac{1}{l_n^{\star}}\right\rceil+\sum_{n\in\bm{N}^{\prime}}\frac{1}{d_n}\\
		&\overset{(c)}{\geq} \left\lceil\sum_{n\in\bar{\bm{N}^{\prime}}}\frac{1}{l_n^{\prime}}\right\rceil+\sum_{n\in\bm{N}^{\prime}}\frac{1}{d_n}{=}K_{WA}.
	\end{split}
\end{equation} 

Inequality (a) comes from the characteristics of average transmission interval. 
Equality (b) is based on the fact that the scheduler constructed by HSI is optimal and achieves full channel utilization for all the occupied channels, thus $\sum_{n\in\bm{N}^{\prime}}\frac{1}{d_n}$ is an integer. Since $\bm{l}^{\prime}$ is the solution of problem \eqref{Aion} for sources in $\bar{\bm{N}^{\prime}}$, then for any $\bar{\bm{N}^{\prime}}$, we have $\lceil\sum_{n\in\bar{\bm{N}^{\prime}}}\frac{1}{l_n^{\prime}}\rceil\leq\lceil\sum_{n\in\bar{\bm{N}^{\prime}}}\frac{1}{l_n^{\star}}\rceil$, then inequality (c) holds. 
If $\lceil\sum_{n\in\bar{\bm{N}^{\prime}}}\frac{1}{l_n^{\prime}}\rceil<\lceil\sum_{n\in\bar{\bm{N}^{\prime}}}\frac{1}{l_n^{\star}}\rceil$, the introduction of HSI leads to better results when all the sources are packed into a single group. Since $K_1$ is the upper bound for the results of HGA, better performance can also be achieved by HSI when the sources in $\bar{\bm{N}^{\prime}}$ are divided into several groups.

Since HGA is a heuristic method for solving problem \eqref{grouping problem}, we can not guarantee that the more sources scheduled by HSI, the better the performance achieved by TGA. Based on the simulations, we find a trick to improve the performance of the proposed method. If the distinct values of AoI constraints of sources in $\bar{\bm{N}^{\prime}}$ after the first part of HSI are the same as that after the second part, we then only schedule the sources found in the first part by HSI. For example, if AoI constraints for the sources in $\bar{\bm{N}^{\prime}}$ after the first part is $\bar{\bm{d}}^{\prime}=[6,6,6,6,6,7,7,9,9,9,9,9,9,9]$. Then after the second part, AoI constraints are $[6,7,7,9,9,9,9]$. 
Such a procedure needs $3$ channels for $\bar{\bm{d}}^{\prime}$. If we use HGA directly for $\bar{\bm{d}}^{\prime}$, we only require 2 channels by dividing them into two groups.

\subsection{Remarks}

Aion, proposed in \cite{liu2021aion}, is a fast algorithm for solving problem \eqref{Aion} and constructing a corresponding scheduler. Compared with Aion, our TGA has the following main differences:
\begin{itemize}
\item Aion is optimal with only consecutively divisible AoI constraints while HSI (the first step of TGA) is optimal for sources with harmonic AoI constraints.
\item Given AoI constraints $\bm{d}$, Aion directly solves problem \eqref{Aion} for $\bm{d}$ assuming all the sources are packed into a same group, while HGA aims at solving our grouping problem \eqref{grouping problem}. HGA divides the sources into different groups, each group can possess different cyclic schedulers and cycle lengths. Since problem \eqref{Aion} is only a special case of \eqref{grouping problem}, the number of channels achieved by Aion will be the upper bound of HGA.
\end{itemize}

Therefore, if HSI is excluded from TGA and $N$ sources are only allowed to be packed into a same group, TGA will be reduced to Aion.

\section{Simulation Results}

In this section, we compare the performance of Aion and TGA. Besides, the number of channels required by consecutively divisible $\bm{l}=[l_1,\cdots,l_N]$ is derived by  $\left\lceil\sum_{n=1}^{N}\frac{1}{l_n}\right\rceil$, not by using CS. Fig.~\ref{aoi10} shows the number of channels achieved by TGA ($\gamma=0.5$), GD, and Aion. We also show the lower bound in Lemma \ref{lbub}. In Fig.~\ref{aoi10}, we vary the number of sources, namely $N$, from 10 to 300 with a step size of $10$. For each $N$, the AoI constraints are generated by following a uniform distribution which takes value in interval $[2,10]$. The number of required channels is averaged over 1000 instances.

The performance achieved by TGA is close to the lower bound. When $N$ increases, the gap between Aion and the lower bound increases from $0.343$ to $14.771$ while the gap for TGA is no more than $0.124$. This phenomenon illustrates that TGA can still ensure good performance when $N$ is large. When $N=300$, the average lower bound is $64.867$ while the average numbers of channels required by Aion and TGA are $79.638$ and $64.961$, respectively. Therefore, with HSI and HGA, the average number of channels decreases by $18.43\%$. This is because Aion can only provide optimal schedulers for consecutively divisible AoI constraints. For the generalized AoI constraints, Aion has to find the optimal consecutively divisible average transmission intervals. Such an operation results in a high update frequency for most sources and requires more channels. Our TGA can not only provide optimal schedulers for harmonic AoI constraints by HSI but also decrease the gap between the average transmission intervals and the AoI constraints of the sources by HGA.

In addition, we find that the number of channels achieved by Aion exceeds GD when the number of sources is larger than $70$. Recall that GD assigns independent channels to the sources with distinct AoI constraints, hence GD allows generalized average transmission interval vectors and achieves the optimal performance for the sources which occupy exactly an integer number of channels. Therefore, when $N$ increases, GD can find more sources that can occupy an integer number of channels and achieve better performance than Aion. 

\begin{figure}[ht]
	
	\centering
	\includegraphics[scale=0.2385]{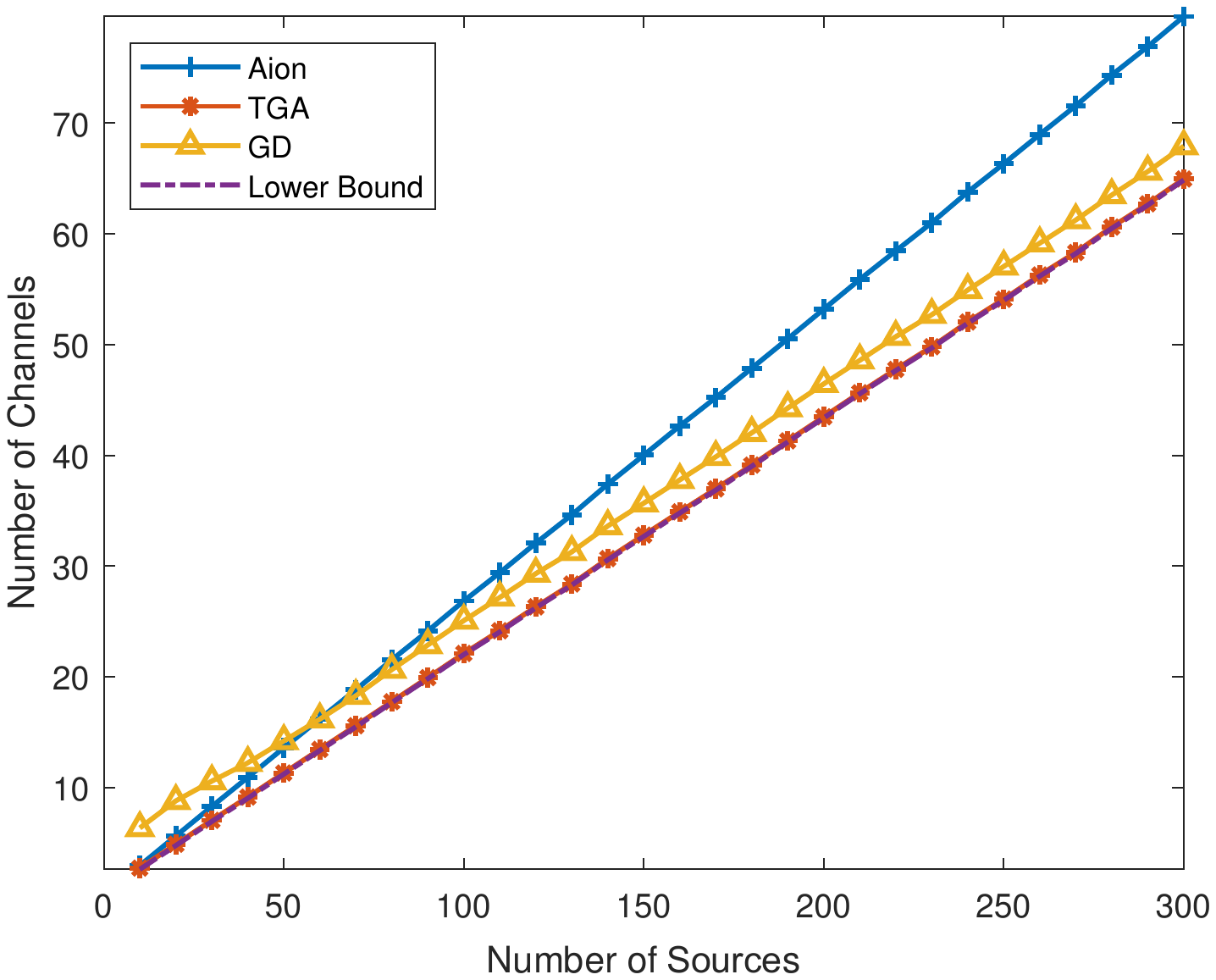}
	\vspace{-0.1in}
	\caption{Performance of Aion, GD and TGA over 1000 instances, where AoI constraints are generated by uniform distribution $[2,10]$.}
	\label{aoi10}
\end{figure}
\begin{figure}[ht]
	
	\centering
	\includegraphics[scale=0.2385]{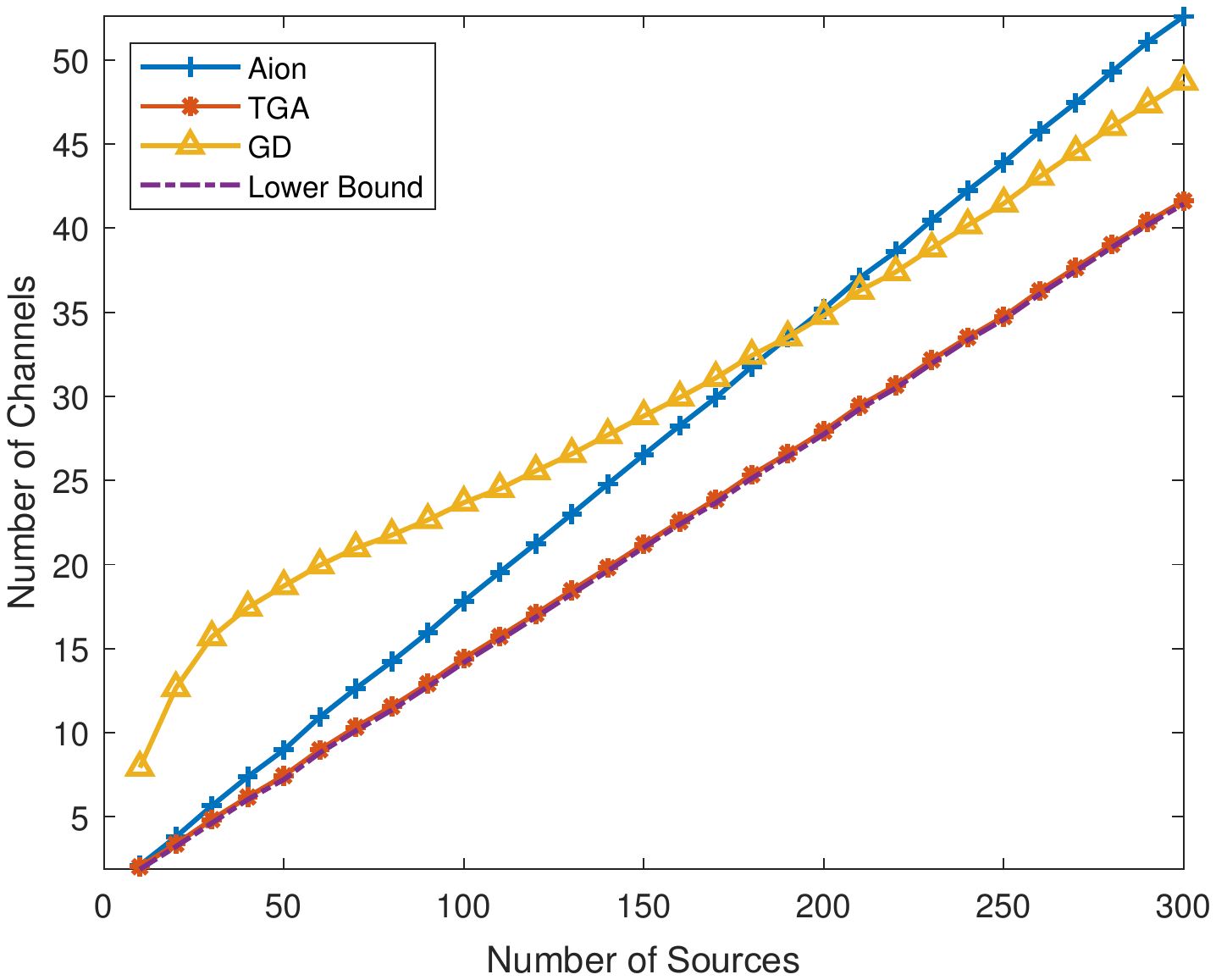}
	\vspace{-0.1in}
	\caption{Performance of Aion, GD and TGA over 1000 instances, where AoI constraints are generated by uniform distribution $[2,20]$.}
	\label{aoi20}
\end{figure}

Fig.~\ref{aoi20} shows the comparison among the three methods mentioned above when the number of distinct values of AoI constraints increases. In Fig.~\ref{aoi20}, 1000 instances of AoI constraint vectors are generated by following a uniform distribution $[2,20]$. When $N$ increases, the gap between Aion and the lower bound increases from $0.241$ to $11.146$ while the gap for TGA is no more than $0.225$. When $N=300$, the average lower bound is $41.439$ while the average numbers of channels required by Aion and TGA are $52.585$ and $41.615$. The gap between TGA and the lower bound is no more than $0.42\%$. Besides, with HSI and HGA, the average number of channels decreases $20.86\%$ compared with Aion.  
Comparing with the results in Fig.~\ref{aoi10}, the gap between Aion and TGA increases when the AoI constraints are more distinct. This is because when the number of distinct values increases, Aion may lead to more sources with average transmission intervals smaller than the AoI constraints and thus require more channels. The introduction of grouping can solve this problem by increasing the solution space of the average transmission interval vectors.

Secondly, the number of sources with the performance of GD lower than Aion increases from $70$ to $210$ when the maximum AoI constraint increases from $10$ to $20$. Given $N$, then the number of sources with the same AoI constraint decreases when the number of distinct values of AoI constraints increases. In this case, GD will waste more channel resources and achieve worse performance. Therefore, when the number of distinct values increases, the performance of GD is better than that of Aion only when $N$ is large. 

\section{Conclusion}

A scheduling problem for minimizing the number of channels with AoI guarantee is studied in this paper. To split the large-scale scheduling problem into multiple small-scale problems, we propose a grouping method which divides the sources into different groups and designs schedulers for each group independently. Therefore, the scheduling problem is then transformed into finding the optimal grouping scheme. A novel two-step grouping algorithm (TGA) is proposed. In the first step, we identify the sources with harmonic AoI constraints and design an optimal scheduler for them. Then for the other sources with generalized AoI constraints, we propose a heuristic grouping algorithm to pack the sources which can be scheduled together with minimum update rates into a same group. In the simulations, TGA shows only $0.42\%$ gap compared with the lower bound of the problem when the number of sources is large.  

\section{ACKNOWLEDGMENT}

This work is sponsored in part by the National Key R\&D Program of China No. 2020YFB1806605, by the Nature Science Foundation of China (No. 62022049, No. 61871254), and Hitachi Ltd.

\bibliography{test}
\bibliographystyle{IEEEtran}

\end{document}